\documentclass[12pt,reqno]{amsart} 
\usepackage{amsmath}
\usepackage{amssymb}
\usepackage{amsthm}
\usepackage{hyperref}
\usepackage{cleveref}
\usepackage{tikz}
\usepackage{caption}
\renewcommand{\geq}{\geqslant}
\renewcommand{\leq}{\leqslant}
\renewcommand{\ge}{\geq}
\renewcommand{\le}{\leq}

\usepackage[colorinlistoftodos,bordercolor=orange,backgroundcolor=orange!20,linecolor=orange,textsize=small,disable]{todonotes}
\newcommand{\rajendra}[1]{\todo[color=blue!20]{RB: #1}}
\newcommand{\stephane}[1]{\todo{SG: #1}}
\newtheorem{thm}{Theorem}
\newtheorem{lem}[thm]{Lemma}
\newtheorem{prop}[thm]{Proposition}

\theoremstyle{definition}
\newtheorem*{definition}{Definition}
\newcommand{\R}{\mathbb{R}}
\newcommand{\dom}{\operatorname{dom}}
\newcommand{\closure}{\operatorname{clo}}
\newcommand{\interior}{\operatorname{int}}
\newcommand{\boundary}{\operatorname{bdry}}
\newcommand{\vertiii}[1]{{\left\vert\kern-0.25ex\left\vert\kern-0.25ex\left\vert
#1 \right\vert\kern-0.25ex\right\vert\kern-0.25ex\right\vert}}

\newcommand{\positivereals}{\mathbb{R}_{++}}
\newcommand{\nonnegativereals}{\mathbb{R}_{+}}

\def \d   {\text {\rm d}}

\def \e   {\text {\rm e}}

\def \min   {\text {\rm min}}

\newcommand{\diag}{\operatorname{diag}}
\def \lim   {\text {\rm lim}}

\def \tr    {\text {\rm tr}}

\def \argmin {\text {\rm argmin}}

\begin{document}

\title[]{Matrix versions of the Hellinger distance}

\author[Rajendra Bhatia]{Rajendra Bhatia}
\address{Ashoka University, Sonepat\\ Haryana, 131029, India}

\email{rajendra.bhatia@ashoka.edu.in}
	
\author[Stephane Gaubert]{Stephane Gaubert}

\address{INRIA and CMAP, Ecole Polytechnique, CNRS, 91128\\ Palaiseau, France}

\email{Stephane.Gaubert@inria.fr}

\author[Tanvi Jain]{Tanvi Jain}
\address{Indian Statistical Institute\\ New Delhi 110016, India}
\email{tanvi@isid.ac.in}

\begin{abstract}
On the space of positive definite matrices we consider distance functions of the form
$d(A,B)=\left[\tr\mathcal{A}(A,B)-\tr\mathcal{G}(A,B)\right]^{1/2},$
where $\mathcal{A}(A,B)$ is the arithmetic mean and $\mathcal{G}(A,B)$ is one of the different versions of the geometric mean.
When $\mathcal{G}(A,B)=A^{1/2}B^{1/2}$ this distance is
$\|A^{1/2}-B^{1/2}\|_2,$ and when $\mathcal{G}(A,B)=(A^{1/2}BA^{1/2})^{1/2}$
it is the Bures-Wasserstein metric.
We study two other cases:
$\mathcal{G}(A,B)=A^{1/2}(A^{-1/2}BA^{-1/2})^{1/2}A^{1/2},$ the Pusz-Woronowicz geometric mean,
and $\mathcal{G}(A,B)=\exp\big(\frac{\log A+\log B}{2}\big),$
the log Euclidean mean.
With these choices $d(A,B)$ is no longer a metric,
but it turns out that $d^2(A,B)$ is a divergence.
We establish some (strict) convexity properties of these divergences.
We obtain characterisations of barycentres
of $m$ positive definite matrices
with respect to these distance measures.
\end{abstract}



\maketitle

\section{Introduction}

Let $p$ and $q$ be two discrete probability distributions;
i.e. $p=(p_1,\ldots,p_n)$ and $q=(q_1,\ldots,q_n)$ are $n$-vectors with nonnegative coordinates
such that $\sum p_i=\sum q_i=1.$
The {\it Hellinger distance} between $p$ and $q$ is the Euclidean norm of the difference between the square roots of $p$ and $q$; i.e.
\begin{equation}
d(p,q)=\|\sqrt{p}-\sqrt{q}\|_2=\left[\sum (\sqrt{p_i}-\sqrt{q_i})^2\right]^{1/2}\!=\left[ \sum(p_i+q_i)-2\sum\sqrt{p_iq_i}\right]^{1/2}\!.\label{eq1}
\end{equation}
This distance and its continuous version, are much used in statistics, where it is customary to take
$d_H(p,q)=\frac{1}{\sqrt{2}}d(p,q)$ as the definition of the Hellinger distance.
We have then 
\begin{equation}
d_H(p,q)=\sqrt{\tr \mathcal{A}(p,q)-\tr\mathcal{G}(p,q)},\label{eq2}
\end{equation}
where $\mathcal{A}(p,q)$ is the arithmetic mean of the vectors $p$ and $q,$
$\mathcal{G}(p,q)$ is their geometric mean, and $\tr\, x$ stands for $\sum x_i.$
\vskip.2in
A matrix/noncommutative/quantum version would seek to replace the probability
vectors $p$ and $q$ by {\it density matrices} $A$ and $B$; i.e.,
positive semidefinite matrices $A,B$ with $\tr\, A=\tr\, B=1.$
In the discussion that follows, the restriction on trace is not needed, and so we let $A$ and $B$ be any two positive semidefinite matrices.
On the other hand, a part of our analysis requires $A$ and $B$ to be positive definite.
This will be clear from the context.
We let $\mathbb{P}$ be the set of $n\times n$ complex positive definite matrices.
The notation $A\ge 0$ means that $A$ is positive (semi) definite.
\vskip.2in
Here we run into the essential difference between the matrix and the scalar case. For 
positive definite matrices $A$ and $B,$
there is only one possible
\todo{SG: minor edit (replaces ``it is clear that'')}
arithmetic mean, $\mathcal{A}(A,B)=(A+B)/2.$
However, the geometric mean $\mathcal{G}(A,B)$ could have different meanings.
Each of these leads to a different version of the Hellinger distance on matrices.
In this paper we study some of these distances and their properties.
\vskip.2in
The Euclidean inner product on $n\times n$ matrices is defined as
$\langle A,B\rangle=\tr\, A^*B.$
The associated {\it Euclidean norm} is 
$$\|A\|_2=(\tr\, A^*A)^{1/2}=(\sum |a_{ij}|^2)^{1/2}.$$
\vskip.2in
Recall that the matrices $AB$ and $BA$ have the same eigenvalues.
Thus if $A$ and $B$ are positive definite, then
$AB$ is not positive definite unless $A$ and $B$ commute.
However, the eigenvalues of $AB$ are all positive as they are the same as the eigenvalues of
$A^{1/2}BA^{1/2}.$
Also every matrix with positive eigenvalues has a unique square root with positive eigenvalues.
If $A,B$ are positive definite,
then we denote by $(AB)^{1/2}$ the square root that has positive eigenvalues.
Since $(AB)^{1/2}=A^{1/2}(A^{1/2}BA^{1/2})^{1/2}A^{-1/2},$
the matrices $(AB)^{1/2}$ and $(A^{1/2}BA^{1/2})^{1/2}$ are similar,
and hence have the same eigenvalues.
\vskip.2in
The straightforward generalisation of \eqref{eq1} for positive definite matrices $A,B$ is 
evidently 
\begin{equation}
d_1(A,B)=\|A^{1/2}-B^{1/2}\|_2=\left[\tr(A+B)-2\tr A^{1/2}B^{1/2}\right]^{1/2}.\label{eq3}
\end{equation}
Another version could be
\begin{equation}
d_2(A,B)=\left[ \tr(A+B)-2\tr(A^{1/2}BA^{1/2})^{1/2}\right]^{1/2}\!=\left[\tr(A+B)-2\tr(AB)^{1/2}\right]^{1/2}\!.\label{eq4}
\end{equation}
\vskip.2in
While it is clear from \eqref{eq3} that $d_1$ is a metric on $\mathbb{P},$
it is not obvious that $d_2$ is a metric.
It turns out that 
\begin{equation}
d_2(A,B)=\min\, \|A^{1/2}-B^{1/2}U\|_2,\label{eq5}
\end{equation}
where the minimum is taken over all unitary matrices $U.$
It follows from this that $d_2$ is a metric.
This is called the {\it Bures distance} in the quantum information
literature and the {\it Wasserstein metric} in the literature on optimal transport.
It plays an important role in both these subjects.
We refer the reader to \cite{bjl} for a recent exposition,
and to \cite{bz,j, mo, t} for earlier work.
The quantity $F(A,B)=\tr (A^{1/2}BA^{1/2})^{1/2}$ is called the {\it fidelity} between the states
$A$ and $B.$
In the special case when $A=uu^*,$ $B=vv^*$
are pure states, we have
$F(A,B)=|u^*v|$ and $d_2(A,B)=\sqrt{2}(1-|u^*v|)^{1/2}.$
For qubit states this is the distance on the Bloch sphere.
\vskip.2in
For various reasons, theoretical and practical,
the most accepted definition of {\it geometric mean} of $A,B$ is the 
entity
\begin{equation}
A\# B=A^{1/2}(A^{-1/2}BA^{-1/2})^{1/2}A^{1/2}.\label{eq6}
\end{equation}
This formula was introduced by Pusz and Woronowicz \cite{pw}.
When $A$ and $B$ commute $A\# B$ reduces to $A^{1/2}B^{1/2}.$
The mean $A\# B$ has been studied extensively for several years and has remarkable properties
that make it useful in diverse areas.
One of them is its connection with operator inequalities related to monotonicity and convexity
theorems for the quantum entropy.
See Chapter 4 of \cite{rbh1} for a detailed exposition.
Another object of interest has been the {\it log Euclidean mean} $\mathcal{L}(A,B)$ defined as
\begin{equation}
\mathcal{L}(A,B)=\exp\left(\frac{\log A+\log B}{2}\right).\label{eq7}
\end{equation}
This mean too reduces to $A^{1/2}B^{1/2}$ when $A$ and $B$ commute,
and has been used in various contexts \cite{xyz}, though it lacks some pleasing properties that $A\# B$ has.
\vskip.2in
Thus it is natural to consider two more matrix versions of the Hellinger distance,
viz, 
\begin{equation}
d_3(A,B)=\left[\tr(A+B)-2\tr(A\# B)\right]^{1/2},\label{eq8}
\end{equation}
and
\begin{equation}
d_4(A,B)=\left[\tr (A+B)-2\tr\mathcal{L}(A,B)\right]^{1/2}.\label{eq9}
\end{equation}
In view of what has been discussed, we may expect that $d_3$ and $d_4$ are metrics on $\mathbb{P}.$
However, it turns out that neither of them obeys the triangle inequality.
Examples are given in Section 2.
Nevertheless, this is compensated by the fact that the squares of $d_3$ and $d_4$ both are {\it divergences},
and hence they can serve as good distance measures.
\vskip.2in

\todo[inline]{SG: inconsistency of notation here $\mathbb{R}_+$ is the set of nonnegative numbers, but later is is used for positive numbers, being the one dimensional version of $\mathbb{P}$. I now use $\mathbb{R}_{++}$ for the positive reals and $\mathbb{R}_+$ for the nonnegative reals, and define it at the first occurrence}
A smooth function $\Phi$ from $\mathbb{P}\times\mathbb{P}$
to the set of nonnegative real numbers, $\nonnegativereals$,
is called a {\it divergence} if
\begin{itemize}
\item[$(i)$] $\Phi(A,B)=0$ if and only if $A=B.$
\item[$(ii)$] The first derivative $D\Phi$ with respect to the second variable
\rajendra{second variable precised}
vanishes on the diagonal;
i.e.,
\begin{equation}
D\Phi(A,X)\vert_{X=A}=0.\label{eq10}
\end{equation}
\item[$(iii)$] The second derivative $D^2\Phi$ is positive on the diagonal; i.e.,
\begin{equation}
D^2\Phi(A,X)\vert_{X=A}(Y,Y)\ge 0\textrm{ for all Hermitian }Y.\label{eq11}
\end{equation}
\end{itemize}
See \cite{am}, Sections 1.2 and 1.3.
\vskip.2in
The prototypical example is the Euclidean divergence $\Phi(A,B)=\|A-B\|_2^2.$
The functions $d_1^2(A,B)$ and $d_2^2(A,B)$ are also divergences.
Another well-known example is the Kullback-Leibler divergence \cite{am}.
A special kind of divergence is the {\it Bregman divergence}
\todo[inline]{SG: the mother function $\varphi$ should take values in $\mathbb{R}$ rather than $\mathbb{R}_+$, think of $\varphi(x)=x\log x$}
corresponding to a strictly convex differentiable
\rajendra{say it is differentiable}
function $\varphi:\mathbb{P}\to\mathbb{R}.$
If $\varphi$ is such a function, then 
\begin{equation}
\Phi(A,B)=\varphi(A)-\varphi(B)-D\varphi(B)(A-B),\label{eq12}
\end{equation}
is called the Bregman divergence corresponding to $\varphi.$
Not every divergence arises in this way.
In particular, $d_H^2(p,q),$ the square of the Hellinger distance,
on probability vectors is not a Bregman divergence.
\vskip.2in
Now we describe our main results.
We will show that both the functions
$$\Phi_3(A,B)=d_3^2(A,B)\textrm{ and }\Phi_4(A,B)=d_4^2(A,B)$$
are divergences.
We will show that $\Phi_3$ and $\Phi_4$ are jointly convex in the variables $A$ and $B,$
and strictly convex in each of the variables separately.
One consequence of this is that for every $m$-tuple
$A_1,\ldots,A_m$ in $\mathbb{P}$ and 
\todo{SG: added ``positive''}
positive weights
$w_1,\ldots,w_m$ the minimisation problem 
\begin{equation}
{\underset{X>0}{\min}}\sum\limits_{j=1}^{m}w_jd^2(X,A_j)\label{eq13}
\end{equation}
has a unique solution when $d=d_3$ or $d_4.$
When $d=d_1$ the minimum in \eqref{eq13} is attained at the $1/2$-power mean
\begin{equation}
Q_{1/2}=\left(\sum\limits_{j=1}^{m}w_jA_j^{1/2}\right)^2.\label{eq13a}
\end{equation}
This is one of the much studied family of classical power means.
When $d=d_2,$ the minimiser in \eqref{eq13} is the
Wasserstein mean \cite{ac, bjl}.
This is the unique solution of the matrix equation
\begin{equation}
X=\sum\limits_{j=1}^{m}w_j(X^{1/2}A_jX^{1/2})^{1/2}.\label{eq13b}
\end{equation}
This mean has major applications in optimal transport, statistics, quantum information and other areas.
Means with respect to various divergences have also been of interest in information theory.
See e.g., \cite{bmdg, nb}.
An inspection of \eqref{eq13a} and \eqref{eq13b} shows a common feature.
Both for $d_1$ and $d_2$ the minimiser in \eqref{eq13} is the solution of the equation
\begin{equation}
X=\sum\limits_{j=1}^{m}w_j\mathcal{G}(X,A_j),\label{eq16a}
\end{equation}
where $\mathcal{G}$ is the version of the geometric mean chosen in the definition of $d.$
That is, $\mathcal{G}(A,B)=A^{1/2}B^{1/2}$ in the case of $d_1,$
and $\mathcal{G}(A,B)=(A^{1/2}BA^{1/2})^{1/2}$ in the case of $d_2.$
It turns out that this is also the case for $d_4$ but not for $d_3.$
When $d=d_3$ the minimisation problem \eqref{eq13} has a unique solution
$X$ which is also the solution of the matrix equation
\begin{equation}
X^2=\frac{2}{\pi}\sum\limits_{j=1}^{m}w_j\int\limits_{0}^{\infty}\left(\lambda X^{-1}+A_j^{-1}\right)^{-2}\sqrt{\lambda}\textrm{d}\lambda.\label{eq17r}
\end{equation}
	This, in general, is different from the solution of the matrix equation
	\begin{equation}
	X=\sum\limits_{j=1}^{m}w_j(X\# A_j).\label{eq13r}
	\end{equation}
When $d=d_4,$ the problem \eqref{eq13} has a unique solution $X$
which is also the solution of the matrix equation
\begin{equation}
X=\sum\limits_{j=1}^{m}w_j\mathcal{L}(X,A_j).\label{eq18a}
\end{equation}

In the past few years there has been extensive work on the Cartan mean
(also known as Karcher or Riemann mean) of positive definite matrices.
\stephane{overfull fixed}
This is the solution of the minimisation problem
\begin{equation}
{\underset{X>0}{\min}}\sum\limits_{j=1}^{m}w_j\delta^2(X,A_j),\label{eq19a}
\end{equation}
where 
\[ \delta(A,B)=\|\log\, A^{-1/2}BA^{-1/2}\|_2
\]
is the Cartan metric on the manifold $\mathbb{P}$.\rajendra{Cartan added}
This mean from classical differential geometry has found several important applications \cite{ba, rbh1, rbh3, fj, n}.
\vskip.2in
Our analysis of $\Phi_4$ leads to some interesting facts about quantum relative entropy.
We observe that the convex function $\varphi(A)=\tr\left(A\log A-A\right)$ leads to the Bregman divergence $\Phi(A,B)=\tr\, A(\log A-\log B)-\tr(A-B),$
and the log Euclidean mean is the barycentre with respect to this Bregman divergence.
As a related issue, we explore
properties of barycentres with respect to general matrix Bregman divergences,
and point out similarities and crucial differences
between the scalar and matrix case.
\vskip.2in
Convexity properties of matrix Bregman divergences have been studied in \cite{bb, pv}, and matrix approximation problems with divergences in \cite{dt}.
Means with respect to matrix divergences are studied in~\cite{chebbimoakher}.
In \cite{sra} Sra studied a related distance function
\[
\delta_S(A,B):= \Big[ \log \det\big(\frac{A+B}{2}\big) - \frac{1}{2}(\log\det A + \log\det B)\Big]^{1/2}
\]
and showed that this is a metric on $\mathbb{P}$. Several parallels between this metric and the Cartan metric are pointed out in \cite{sra}.

\vskip.3in

\section{Convexity and derivative computations}
\vskip.2in
Inequalities for traces of matrix expressions have a long history.
For the different geometric means mentioned in Section 1,
we know \cite{bg} that 
\begin{equation}
\tr(A\# B)\le \tr\mathcal{L}(A,B)\le \tr(A^{1/2}B^{1/2})\le\tr(AB)^{1/2}.\label{eq14}
\end{equation}
It follows that
\begin{equation}
d_3^2(A,B)\ge d_4^2(A,B)\ge d_1^2(A,B)\ge d_2^2(A,B).\label{eq15}
\end{equation}
Since $d_1$ is a metric, this implies that $d_3^2(A,B)=0$ if and only if $A=B.$
The same is true for $d_4^2(A,B).$
Thus $\Phi_3$ and $\Phi_4$ satisfy the first condition in the definition of a divergence.
To prove $\Phi_3$ is a divergence we need to compute its first and second derivatives.
These results are of independent interest.

\begin{prop}\label{prop-new2}
Let $A$ be a positive definite matrix.
Let $g$ be the map on $\mathbb{P}$ defined as
$$g(X)=A\# X.$$
Then the derivative of $g$ is given by the formula
\begin{equation}
Dg(X)(Y)=\int\limits_{0}^{\infty}(\lambda+XA^{-1})^{-1}Y(\lambda+A^{-1}X)^{-1}\d\nu(\lambda),\label{eq5b}
\end{equation}
where $\d\nu(\lambda)=\frac{1}{\pi}\lambda^{1/2}\d\lambda.$
\end{prop}

\begin{proof}
We will use the integral representation
	\begin{equation}
	x^{1/2}=\frac{1}{\sqrt{2}}+\int\limits_{0}^{\infty}\left(\frac{\lambda}{\lambda^2+1}-\frac{1}{\lambda+x}\right)\d\nu(\lambda),\label{eq2b}
	\end{equation}
	where $\d\nu(\lambda)=\frac{1}{\pi}\lambda^{1/2}\d\lambda.$
	See \cite{rbh} p.143.
	Using this we see that the derivative of the function $X\to X^{1/2}$ is the linear map
	\begin{equation}
	DX^{1/2}(Y)=\int\limits_{0}^{\infty}(\lambda+X)^{-1}Y(\lambda+X)^{-1}\d\nu(\lambda),\label{eq4b}
	\end{equation}
	where $Y$ is any Hermitian matrix.
This shows that	\begin{align*}
	 & Dg(X)(Y)\\
	 & = \int\limits_{0}^{\infty}A^{1/2}(\lambda+A^{-1/2}XA^{-1/2})^{-1}A^{-1/2}YA^{-1/2}(\lambda+A^{-1/2}XA^{-1/2})^{-1}A^{1/2}\d\nu(\lambda)\nonumber\\
	 & = \int\limits_{0}^{\infty}(\lambda+XA^{-1})^{-1}Y(\lambda+A^{-1}X)^{-1}\d\nu(\lambda).
	\end{align*}
	This proves the proposition.
	\end{proof}

\begin{thm}\label{thm1}
Let $D\Phi_3$ and $D^2\Phi_3$ be the first and the second derivatives of $\Phi_3.$
Then 
\begin{equation}
D\Phi_3(A,A)=0,\label{eq16}
\end{equation}
\begin{equation}
D^2\Phi_3(A,A)(Y,Y)=\frac{1}{2}\tr\, YA^{-1}Y.\label{eq17}
\end{equation}
(In other words, the gradient of $\Phi_3$ at every diagonal point is $0$ and the Hessian is positive.)
\end{thm}
\begin{proof}
For a fixed $A,$ let $g$ be the map on $\mathbb{P}$ defined as
$g(X)=A\# X.$
When $X=A,$ the expression in \eqref{eq5b} reduces to
$$\frac{1}{\pi}\int\limits_{0}^{\infty}\frac{\lambda^{1/2}}{(1+\lambda)^2}\d\lambda\, Y=\frac{1}{2}Y.$$
Recalling that $\Phi_3(A,X)=\tr(A+X)-2\tr g(X),$
we see that 
$$D\Phi_3(A,X)\vert_{X=A}(Y)=0\textrm{ for all }Y.$$
This establishes \eqref{eq16}.
Next note that for the second derivative we have
\begin{equation}
D^2\Phi_3(A,X)(Y,Z)=-2D^2 \left(\tr g(X)\right)(Y,Z).\label{eq21}
\end{equation}
From \eqref{eq5b} we see that
\begin{align}
 & D\left(\tr\, g(X)\right)(Y)\nonumber\\
& =\int\limits_{0}^{\infty}\tr(\lambda+XA^{-1})^{-1}Y(\lambda+A^{-1}X)^{-1}\d\nu(\lambda)).\label{eq22}
\end{align}
By definition 
$$D^2(\tr\, g(X))(Y,Z)=\frac{\d}{\d t}\vert_{t=0}D(\tr\, g(X+tZ))(Y).$$
Hence, from \eqref{eq22} we see that $D^2(\tr\, g(X))(Y,Z) $
is equal to\rajendra{displayed eq edited}
\begin{align}
& -\int\limits_{0}^{\infty}\tr(\lambda+XA^{-1})^{-1}ZA^{-1}(\lambda+XA^{-1})^{-1}Y(\lambda+A^{-1}X)^{-1}\d\nu(\lambda)\nonumber\\
& \ \ -\int\limits_{0}^{\infty}\tr(\lambda+XA^{-1})^{-1}Y(\lambda+A^{-1}X)^{-1}A^{-1}Z(\lambda+A^{-1}X)^{-1}\d\nu(\lambda).\label{eq23}
\end{align}
When $X=A$ and $Z=Y,$ this reduces to give 
\begin{eqnarray*}
D^2\Phi_3(A,A)(Y,Y) & = & \frac{2}{\pi}\int\limits_{0}^{\infty}\frac{\lambda^{1/2}}{(1+\lambda)^3}\d\lambda\, \tr\, YA^{-1}Y\\
 & = & \frac{1}{2}\tr\, YA^{-1}Y.
\end{eqnarray*}
This proves \eqref{eq17}.
\end{proof}
\todo[inline]{SG: rewrote what follows as the restriction to $\mathbb{R}_+$ is unnatural and the notation $\positivereals$ was undefined}
Consider maps $f$ defined on $\mathbb{P}$ and taking
values in $\mathbb{P}$ or $\positivereals$ (the set of positive real numbers).
We say that $f$
is {\it concave} if for all $X,Y$ in $\mathbb{P}$
and $0\le \alpha\le 1$
\begin{equation}
f((1-\alpha)X+\alpha Y)\ge (1-\alpha)f(X)+\alpha f(Y).\label{eq24}
\end{equation}
It is {\it strictly concave} if the two sides of \eqref{eq24} are equal only if $X=Y.$
A map $f$ from $\mathbb{P}\times\mathbb{P}$ into $\mathbb{P}$ or $\mathbb{R}_+$ is called {\it jointly concave} if for all
$X_1,X_2,Y_1,Y_2$ in $\mathbb{P}$ and $0\le \alpha\le 1,$
$$ f((1-\alpha)X_1+\alpha Y_1,(1-\alpha)X_2+\alpha Y_2)\ \ge (1-\alpha)f(X_1,X_2)+\alpha f(Y_1,Y_2).$$

It is a basic fact in the theory of the geometric mean that $A\# B$ is jointly concave in $A$ and $B$, see~\cite{ando,alm}. However, it is not strictly jointly concave. Indeed, even the function $f(a,b)=\sqrt{ab}$ on $\nonnegativereals\times\nonnegativereals$ is not strictly jointly concave (its restriction to the diagonal is linear). Our next theorem says that in each of the variables separately, the geometric mean is strictly concave.
\rajendra{para above edited}

\begin{thm}\label{thm2}
For each $A$ the function
$$f(X)=\tr\, A\# X$$
is strictly concave on $\mathbb{P}.$ This implies that the function
$g(X)=A\# X$ is also strictly concave.\rajendra{last sentence added}
\end{thm}

\begin{proof}
Suppose 
$$\tr\left(A\#\left(\frac{X+Y}{2}\right)\right)=\frac{\tr\, A\# X+\tr\, A\# Y}{2}.$$
We have to show that this implies $X=Y.$
Rewrite the above equality as
$$\tr\left\{A\#\left(\frac{X+Y}{2}\right)-\frac{A\# X+A\# Y}{2}\right\}=0.$$
By the concavity of $A\# X,$ the expression inside the braces is positive semidefinite.
The trace of such a matrix is zero if and only if the matrix itself is zero.
Hence 
$$A\#\left(\frac{X+Y}{2}\right)=\frac{A\# X+A\# Y}{2}.$$
Using the definition \eqref{eq6} this can be written as 
\begin{eqnarray*}
A^{1/2}\left(A^{-1/2}\frac{X+Y}{2}A^{-1/2}\right)^{1/2}A^{1/2} & = & \frac{1}{2}A^{1/2}\left(A^{-1/2}XA^{-1/2}\right)^{1/2}A^{1/2}\\
 &  & +\frac{1}{2}A^{1/2}(A^{-1/2}YA^{-1/2})^{1/2}A^{1/2}.
\end{eqnarray*}
Cancel the factors $A^{1/2}$ occurring on both sides,
then square both sides,
and rearrange terms to get
\begin{eqnarray*}
A^{-1/2}(X+Y)A^{-1/2}-(A^{-1/2}XA^{-1/2})^{1/2}(A^{-1/2}YA^{-1/2})^{1/2} & & \\
\ -(A^{-1/2}YA^{-1/2})^{1/2}(A^{-1/2}XA^{-1/2})^{1/2} & = & 0.
\end{eqnarray*}
This is the same as saying 
$$\left[(A^{-1/2}XA^{-1/2})^{1/2}-(A^{-1/2}YA^{-1/2})^{1/2}\right]^2=0.$$
The square of a Hermitian matrix $Z$ is zero only if $Z=0.$
Hence, we have
$$(A^{-1/2}XA^{-1/2})^{1/2}=(A^{-1/2}YA^{-1/2})^{1/2}.$$
From this it follows that $X=Y.$

Finally, if $X,Y$ are to elements of $\mathbb{P}$ such that $g((X+Y)/2)=(g(X)+g(Y))/2$, taking traces on both sides, we have, $f((X+Y)/2)=(f(X)+f(Y))/2.$ We have seen that this implies $X=Y$. \rajendra{last para added}
\end{proof}

As a consequence, we observe that 
$$\Phi_3(A,B)=\tr(A+B)-2\tr(A\# B)$$
is jointly convex in $A$ and $B$ and is strictly convex in each of the variables separately.
\vskip.2in
Now we turn to the analysis of $\Phi_4$ on the same lines as above.
The arguments we present in this case are quite different.
From \eqref{eq15} we know that 
$$\Phi_3(A,B)\ge\Phi_4(A,B)\ge \Phi_1(A,B).$$
We also know that 
$$\Phi_3(A,A)=\Phi_4(A,A)=\Phi_1(A,A)=0,$$
and
$$D\Phi_1(A,A)=D\Phi_3(A,A)=0.$$
Together, these three relations lead to the conclusion that 
$$D\Phi_4(A,A)=0.$$
Thus $\Phi_4$ satisfies condition \eqref{eq10}.
\vskip.2in
By a theorem of Bhagwat and Subramanian \cite{bs}
\begin{equation}
\exp\left(\frac{1}{m}\sum\limits_{j=1}^{m}\log\, A_j\right)={\underset{p\to 0^+}{\lim}}\left(\frac{1}{m}\sum\limits_{j=1}^{m}A_j^p\right)^{1/p}.\label{eq25}
\end{equation}
One of the several remarkable concavity theorems of Carlen and Lieb, \cite{cl1,cl2} says that
the expression $\tr\left(\sum A_j^p\right)^{1/p}$ is jointly concave in $A_1,\ldots,A_m,$
when $0<p\le 1,$ and jointly convex when $1\le p\le 2.$
Using equation \eqref{eq25} we obtain from this the joint concavity of $\tr\mathcal{L}(A,B).$
As a consequence $\Phi_4(A,B)$ is jointly convex in $A,B.$
Hence we have proved the following theorem.

\begin{thm}
The function $\Phi_4$ is a divergence on $\mathbb{P}.$
\end{thm}

We have shown that $\Phi_3$ and $\Phi_4$ are divergences. But unlike $\Phi_1$ and $\Phi_2$ they are not the squares of metrics on $\mathbb{P},$
i.e., $d_3$ and $d_4$ are not metrics. The following two examples show that $d_3$ and $d_4$ do not satisfy the triangle inequality.
\vskip.2in
Let
$$A=\begin{bmatrix} 2 & 5\\ 5 & 17\end{bmatrix},\ 
B=\begin{bmatrix} 13 & 8 \\ 8 & 5 \end{bmatrix},\ 
C=\begin{bmatrix} 5 & 3 \\ 3 & 10 \end{bmatrix}.$$
Then $d_3(A,B)\approx 5.0347$ and $d_3(A,C)+d_3(C,B)\approx 4.6768.$
This example is a small modification of one suggested to us by Suvrit Sra, to whom we are thankful.
\vskip.1in
Let
$$A=\begin{bmatrix} 4 & -7\\ -7 & 13\end{bmatrix},B=\begin{bmatrix} 8 & -2 \\ -2 & 1 \end{bmatrix},C=\begin{bmatrix} 5 & -4 \\ -4 & 5\end{bmatrix}.$$
Then $d_4(A,B)\approx 3.3349$ and $d_4(A,C)+d_4(C,B)\approx 3.3146.$
\vskip.2in
Next we study some more properties of $\Phi_4$, like its strict convexity in each of the arguments,
and its connections with matrix entropy.
To put these in context we recall some facts about Bregman divergence.
\vskip.2in
Let $\varphi:\mathbb{R}_+\to\mathbb{R}$ be a smooth strictly convex function
and let\todo{SG: I now assume the mother function $\varphi$ takes value in $\mathbb{R}$ rather than in $\mathbb{R}_+$ [to be done systematically]]}
\begin{equation}
\Phi(x,y)=\varphi(x)-\varphi(y)-\varphi^\prime(y)(x-y),\label{eq26}
\end{equation}
be the associated Bregman divergence.
Then $\Phi$ is strictly convex in the variable $x$ but need not be convex in $y.$
(See, e.g., \cite{dt} Section 2.2.)

Given $x_1,\ldots,x_m$ in $\mathbb{R}_+,$ the minimiser
\begin{equation}
\argmin\sum\limits_{j=1}^{m}\frac{1}{m}\Phi(x_j,x),\label{eq27}
\end{equation}
always turns out to be the arithmetic mean 
$$\overline{x}=\sum\limits_{j=1}^{m}\frac{1}{m}x_j,$$
independent of the mother function $\varphi.$
\vskip.2in
In fact, this property characterises Bregman divergences;
see \cite{dt,bmdg}.
We can also consider the problem 
\begin{equation}
\argmin\sum\limits_{j=1}^{m}\frac{1}{m}\Phi(x,x_j).\label{eq28}
\end{equation}
In this case, a calculation shows that the solution is the quasi-arithmetic mean (the Kolmogorov mean) associated
with the function $\varphi^\prime.$
More precisely, the solution of \eqref{eq28}, which we may think
of as the mean, or the barycentre, of the points $x_1,\ldots,x_m$ with respect to the divergence $\Phi$ is
\begin{equation}
\mu_\Phi(x_1,\ldots,x_m)={\varphi^\prime}^{-1}\left(\sum\limits_{j=1}^{m}\frac{1}{m}\varphi^\prime(x_j)\right).\label{eq29}
\end{equation}
\vskip.2in
We wish to study the matrix version of the problems~\eqref{eq27} and~\eqref{eq28}.
Here we run into a basic difference between the one-variable and the several-variables cases. It is natural to replace the derivative $\varphi^\prime$
in~\eqref{eq29} by the gradient $\nabla\varphi$ in the several-variables
case. If $\varphi$ is a differentiable strictly convex function defined on an open interval $I$ of $\R$, then, its derivative $\varphi^\prime$ is a strictly monotone continuous function, and hence a homeomorphism from $I$ to its image $\varphi^{\prime}(I)$. In particular, $(\varphi^{\prime})^{-1}$ is defined. The appropriate
generalisation of these facts to the several-variable case requires the
notion of a {\em Legendre type} function.
\begin{definition}[Section 26 in~\cite{rockafellar} or Def.\ 2.8 in~\cite{bb97}]
Suppose $\varphi$ is a convex lower-semicontinuous function from $\mathbb{R}^n$ to $\mathbb{R}\cup\{+\infty\}$, and let $\dom f:= \{x\in\mathbb{R}^n\mid \varphi(x)<+\infty\}$.
We say that $\varphi$ is 
of {\em Legendre type} if it satisfies
\begin{enumerate}\renewcommand{\theenumi}{\roman{enumi}}
\item $\interior \dom\varphi\neq \varnothing$,
\item $\varphi$ is differentiable on $\interior \dom\varphi$,
\item $\varphi$ is strictly convex on $\interior \dom\varphi$, 
\item\label{itdef-iv} $\lim_{t\to 0^+} \langle \nabla \varphi(x+t(y-x)),y-x \rangle  =-\infty$, for
all $x\in \boundary(\dom(\varphi))$ and $y\in 
\interior \dom\varphi$.
\end{enumerate}
\end{definition}
If $\varphi$ is of Legendre type, the gradient mapping
$\nabla \varphi$ is a homeomorphism from $\interior \dom\varphi$ to $\interior \dom\varphi^\star$, where $\varphi^\star$ denotes
the Legendre-Fenchel conjugate of $\varphi$. See Theorem~26.5 in~\cite{rockafellar}. 

\begin{lem}\label{lem-new}
If $\varphi$ is of Legendre type, and $\Phi$ is the Bregman divergence associated
with $\varphi$, and $a_1,\dots,a_m\in \interior \dom\varphi$, then the function
\[ x\mapsto \sum_{j=1}^m \Phi(x,a_j)
\]
achieves its minimum at a unique point, which belongs to $\interior \dom \varphi$. 
\end{lem}
The proof is given in \Cref{appendixa}.
We shall apply this lemma in the situation where $\varphi$ is a convex function defined only on $\mathbb{P}$ and taking finite values on this set. 
The map $\varphi$ trivially extends to a convex lower-semicontinuous function defined on the whole space of Hermitian matrices---set $\varphi(X):=\liminf_{Y\to X,\; Y\in \mathbb{P}} \varphi(Y)$ for $X\in\boundary(\mathbb{P})$, and $\varphi(X)=+\infty$ if $X\not\in \boundary(\mathbb{P})$. We shall say that
the original function $\varphi$ defined on $\mathbb{P}$ is of Legendre type if its extension is of Legendre type. 
\begin{thm}\label{thm3}
Let $\varphi$ be a differentiable strictly convex function from $\mathbb{P}$ to 
$\mathbb{R},$
and let $\Phi$ be the Bregman divergence corresponding to $\varphi.$ Then:
\begin{enumerate}\renewcommand{\theenumi}{\roman{enumi}}
\item\label{it-i} The minimiser in the problem
\begin{equation}
\argmin_{X\in \mathbb{P}}\sum\limits_{j=1}^{m}\frac{1}{m}\Phi(A_j,X),\label{eq30}
\end{equation}
is the arithmetic mean
$\sum\limits_{j=1}^{m}\frac{1}{m}A_j.$
\item\label{it-ii} If, in addition, $\varphi$ is of Legendre type, then the problem
\begin{equation}
\argmin_{X\in\mathbb{P}}\sum\limits_{j=1}^{m}\frac{1}{m}\Phi(X,A_j)\label{eq31}
\end{equation}
has a unique solution, and this is given by
\begin{equation}
X=(\nabla \varphi)^{-1}\Big(\sum\limits_{j=1}^{m}\frac{1}{m}\nabla \varphi(A_j)\Big) \enspace .\label{eq32}
\end{equation}
\item\label{it-iii} If $\psi$ is any differentiable strictly convex function from $\positivereals$ to $\mathbb{R}$ and $\Phi$ is
the Bregman divergence on $\mathbb{P}$ corresponding to
the function $\varphi(X):=\tr\psi(X)$ on $\mathbb{P}$,
then the solution of the minimisation problem~\eqref{eq31} is
\begin{equation}
X= (\psi^{\prime})^{-1}\Big(\sum\limits_{j=1}^{m}\frac{1}{m}\psi^\prime(A_j)\Big)
\enspace .\label{eq32a}
\end{equation}
\end{enumerate}
\end{thm}
\begin{proof}
(\ref{it-i}). Since $\Phi$ is given by \eqref{eq12},
\begin{eqnarray*}
\sum\limits_{j=1}^{m}\frac{1}{m}\Phi(A_j,X) & = & \sum\limits_{j=1}^{m}\frac{1}{m}\varphi(A_j)-\varphi(X)-\sum\limits_{j=1}^{m}\frac{1}{m}D\varphi(X)(A_j-X)\\
 & = & \sum\limits_{j=1}^{m}\frac{1}{m}\varphi(A_j)-\varphi(X)-D\varphi(X)\left(\sum\limits_{j=1}^{m}\frac{1}{m}A_j-X\right)\\
 & = & \sum\limits_{j=1}^{m}\frac{1}{m}\varphi(A_j)-\varphi(X)-D\varphi(X)(\overline{A}-X),
\end{eqnarray*}
where $\overline{A}$ denotes the arithmetic mean $\sum\limits_{j=1}^{m}\frac{1}{m}A_j.$ Hence
$$\sum\limits_{j=1}^{m}\frac{1}{m}\Phi(A_j,\overline{A})=\sum\limits_{j=1}^{m}\frac{1}{m}\varphi(A_j)-\varphi(\overline{A}).$$
Since $\varphi$ is strictly convex, for every $X\ne\overline{A}$
$$\varphi(\overline{A})-\varphi(X)>D\varphi(X)(\overline{A}-X).$$
This implies that
$$\sum\limits_{j=1}^{m}\frac{1}{m}\Phi(A_j,X)>\sum\limits_{j=1}^{m}\frac{1}{m}\Phi(A_j,\overline{A})$$
which shows that $\overline{A}$ is the unique minimiser of the problem \eqref{eq30}.
\vskip.1in
(\ref{it-ii}). Let $\Psi$ be the map from $\mathbb{P}$ to $\mathbb{R}_+$ defined as
$$\Psi(X)=\sum\limits_{j=1}^{m}\frac{1}{m}\Phi(X,A_j).$$
Then
$$D\Psi(X)(Z)=D\varphi(X)(Z)-\sum\limits_{j=1}^{m}\frac{1}{m}D\varphi(A_j)(Z).$$
\Cref{lem-new} shows that the minimum of the map $\Psi$ on the set $\mathbb{P}$
is achieved at some point $X \in \mathbb{P}$, and by the first order
optimality condition, $D\Psi(X)=0$, showing that $X$ satisfies~\eqref{eq32}.

\vskip.1in

(\ref{it-iii}). If $\psi$ is a differentiable convex function on $\positivereals$
and $\Phi$ is the Bregman divergence corresponding to $\varphi=\tr\psi,$
then $\nabla\varphi(X)=\psi^\prime(X).$ Hence, to show that the
minimisation problem \eqref{eq31} has a solution, 
it suffices to show that the first order optimality condition
\begin{align}
 \psi^\prime(X) = \sum\limits_{j=1}^{m}\frac{1}{m}\psi^\prime(A_j)
\label{e-lasteq}
\end{align}
is satisfied for some $X$ in $\mathbb{P}$. Since $\psi$ is strictly convex,
as noted above, $\psi'$ is strictly increasing and is a homeomorphism from $\positivereals$ to the interval $J:=\psi'(\positivereals)$. 
The spectrum of each matrix $\psi'(A_j)$ belongs to $J$, and so
the spectrum of $\sum\limits_{j=1}^{m}\frac{1}{m}\psi^\prime(A_j)$
also belongs to $J$, which implies that~\eqref{e-lasteq} is solvable.
\end{proof}
The assumption that $\varphi$ is of Legendre type is not needed
in the tracial case (statement~(\ref{it-iii})). \Cref{th-cex}
in \Cref{appendixb} shows that this assumption cannot be dispensed
with in the case of statement~(\ref{it-ii}). \stephane{Para added}

The much studied convex function 
\begin{equation}
\varphi(x)=x\log x-x,\label{eq33}
\end{equation}
on $\mathbb{R}_+$ leads to the Bregman divergence
\begin{equation}
\Phi(x,y)=x(\log x-\log y)-(x-y).\label{eq34}
\end{equation}
This is called the {\it Kullback-Leibler divergence}.
Since $\varphi^\prime(x)=\log x,$
the solution of the minimisation problem \eqref{eq28} in this case is
\begin{equation*}
\mu_\Phi(x_1,\ldots,x_m)=\exp\left(\frac{1}{m}\sum\limits_{j=1}^{m}\varphi(x_j)\right)=\prod\limits_{j=1}^{m}x_j^{1/m},
\end{equation*}
the geometric mean of $x_1,\ldots,x_m.$
\vskip.2in
As a matrix analogue of \eqref{eq33} one considers the function on $\mathbb{P}$ defined as
\begin{equation}
\varphi(A)=\tr(A\log A-A).\label{eq35}
\end{equation}
The associated Bregman divergence then is
\begin{equation}
\Phi(A,B)=\tr\, A(\log A-\log B)-\tr(A-B).\label{eq36}
\end{equation}
(See \cite{am}, p.12).
The quantity 
\begin{equation}
S(A|B)=\tr\, A(\log A-\log B),\label{eq37}
\end{equation}
is called the {\it relative entropy}
and has been of great interest in quantum information.
Given $A_1,\ldots,A_m$ in $\mathbb{P},$
their barycentre with respect to the divergence $\Phi,$
i.e., the solution of the minimisation problem \eqref{eq31} is the log Euclidean mean
\begin{equation}
\mathcal{L}(A_1,\ldots,A_m)=\exp\left(\frac{1}{m}\sum\limits_{j=1}^{m}\log A_j\right).\label{eq38}
\end{equation}
\vskip.2in
It is also of interest to compute the {\it variance} of the points $A_1,\ldots,A_m$
with respect to $\Phi,$ i.e.,
the minimum value of the objective function in \eqref{eq31}.
This is the quantity 
\begin{equation}
\sigma_\Phi^2=\sum\limits_{j=1}^{m}\frac{1}{m}\Phi(\mu_\Phi, A_j).\label{eq39}
\end{equation}
For the divergence $\Phi$ in \eqref{eq36}, $\mu_\Phi$ is the log Euclidean mean $\mathcal{L}$ given in \eqref{eq38}.
So 
\begin{eqnarray*}
\sigma_\Phi^2 & = & \frac{1}{m}\sum\limits_{j=1}^{m}\Phi(\mathcal{L}, A_j)\\
 & = & \frac{1}{m}\sum\limits_{j=1}^{m}\left[\tr\mathcal{L}(\log\mathcal{L}-\log A_j)-\tr(\mathcal{L}-A_j)\right]\\
 & = & \frac{1}{m}\tr\left\{\sum\limits_{j=1}^{m}\left[\mathcal{L}\left(\frac{1}{m}\sum\limits_{k=1}^{m}\log A_k-\log A_j\right)-(\mathcal{L}-A_j)\right]\right\}\\
 & = & -\tr\mathcal{L}+\frac{1}{m}\tr\sum\limits_{j=1}^{m} A_j.
\end{eqnarray*}
In other words 
\begin{equation}
\sigma_\Phi^2=\tr\mathcal{A}(A_1,\ldots,A_m)-\tr\mathcal{L}(A_1,\ldots,A_m),\label{eq40}
\end{equation}
the difference between the traces of the arithmetic and the log Euclidean means of $A_1,\ldots,A_m.$
\vskip.2in
In particular, the divergence $\Phi_4(A,B)$ can be characterised using \eqref{eq40},
as the minimum value 
\begin{equation}
{\underset{X>0}{\min}}\left[\Phi(X,A)+\Phi(X,B)\right],\label{eq41}
\end{equation}
where $\Phi$ is defined by \eqref{eq36}.
Using this characterisation we can show that the function $\Phi_4(A,B)$ is strictly convex in each of the variables separately.  To this end,  we recall
the following lemma of convex analysis, showing that the ``marginal''
of a jointly convex function is convex; compare with Proposition 2.22 
of~\cite{rockafellarandwets} where a similar result (without the strictness conclusion) is provided.
\todo[inline]{SG: added last sentence with a ref to \cite{rockafellarandwets} as this is known in convex analysis}

\begin{lem}\label{lem4}
Let $f(x,y)$ be a jointly convex function which is strictly convex in each of its variables separately.
Suppose for each $a,b$
\begin{equation}
g(a,b)={\underset{x}{\min}}\left[ f(x,a)+f(x,b)\right],\label{eq42}
\end{equation}
exists. Then the function $g(a,b)$ is jointly convex, and is strictly convex in each of the variables separately.
\end{lem}

\begin{proof}
Given $a_1,a_2,b_1,b_2,$
choose $x_1$ and $x_2$ such that 
$$g(a_1,b_1)=f(x_1,a_1)+f(x_1,b_1)$$
and
$$g(a_2,b_2)=f(x_2,a_2)+f(x_2,b_2).$$
Then
\begin{align*}
 & g\left(\frac{a_1+a_2}{2},\frac{b_1+b_2}{2}\right)\\
 & \ \ \le f\left(\frac{x_1+x_2}{2},\frac{a_1+a_2}{2}\right)+f\left(\frac{x_1+x_2}{2},\frac{b_1+b_2}{2}\right)\\
 & \ \ \le \frac{1}{2}\left[f(x_1,a_1)+f(x_2,a_2)+f(x_1,b_1)+f(x_2,b_2)\right]\\
 & \ \ = \frac{1}{2}\left[g(a_1,b_1)+g(a_2,b_2)\right].
\end{align*}
This shows that $g$ is jointly convex.
Now we show that it is strictly convex in the first variable.
\vskip.1in
Let $a_1,a_2,b$ be any three points with $a_1\ne a_2.$
Choose $x_1$ and $x_2$ such that
$$g(a_1,b)=f(x_1,a_1)+f(x_1,b)$$
and
$$g(a_2,b)=f(x_2,a_2)+f(x_2,b).$$
Two cases arise.
If $x_1=x_2=x,$ then
\begin{eqnarray*}
f\left(\frac{x_1+x_2}{2}, \frac{a_1+a_2}{2}\right) & = & f\left(x,\frac{a_1+a_2}{2}\right)\\
 & < & \frac{1}{2}\left[f(x,a_1)+f(x,a_2)\right],
\end{eqnarray*}
because of strict convexity of $f$ in the second variable.
This implies that
\begin{eqnarray*}
g\left(\frac{a_1+a_2}{2},b\right) & < &\frac{1}{2}\left[f(x,a_1)+f(x,a_2)+f(x,b)+f(x,b)\right]\\
 & = & \frac{1}{2}\left[g(a_1,b)+g(a_2,b)\right].
\end{eqnarray*}
If $x_1\ne x_2,$ then by strict convexity of $f$ in the first variable,
$$f\left(\frac{x_1+x_2}{2},b\right)<\frac{1}{2}\left[f(x_1,b)+f(x_2,b)\right],$$
and by joint convexity of $f$
$$f\left(\frac{x_1+x_2}{2},\frac{a_1+a_2}{2}\right)\le \frac{1}{2}\left[f(x_1,a_1)+f(x_2,a_2)\right].$$
Adding the last two inequalities we get
$$g\left(\frac{a_1+a_2}{2},b\right)<\frac{1}{2}\left[g(a_1,b)+g(a_2,b)\right].$$
Thus $g(a,b)$ is strictly convex in the first variable,
and by symmetry it is so in the second variable.
\end{proof}

\begin{thm}\label{thm5}
For each $A,$ the function $f(X)=\Phi_4(X,A)$
is strictly convex on $\mathbb{P}.$
\end{thm}

\begin{proof}
One of the fundamental, and best known,
properties of the relative entropy $S(A|B)$ is that it is jointly convex
function of $A$ and $B.$
(See, e.g., Section IX.6 in \cite{rbh}.)
It is also known that if $\varphi$ is strictly convex function on $\mathbb{R}_+,$
then the function $\tr\, \varphi(X)$ is strictly convex on $\mathbb{P}.$
(See, e.g., Theorem 4 in \cite{bjl2}.)
It follows from this that $S(A|B)$ is strictly convex in each of the variables separately.
Combining these properties of $S(A|B),$
Lemma \ref{lem4} and the characterisation of $\Phi_4(A,B)$ as the
minimum value in \eqref{eq41} we obtain Theorem \ref{thm5}.
\end{proof}

It might be pertinent to add here that the question of equality in the joint convexity inequality 
\begin{equation}
S\left(\frac{A_1+A_2}{2}\big|\frac{B_1+B_2}{2}\right)\le\frac{S(A_1|B_1)+S(A_2|B_2)}{2},\label{eq43}
\end{equation}
has been addressed in \cite{hbmp} and \cite{jr}.
In \cite{jr} Jencova and Ruskai show that the equality holds in \eqref{eq43} if and only if
\begin{eqnarray*}
\log(A_1+A_2)-\log(B_1+B_2) & = & \log A_1-\log B_1\\
  & = & \log A_2-\log B_2.
	\end{eqnarray*}
	On the other hand, Hiai et al \cite{hbmp}
	show that equality holds in \eqref{eq43} if and only if
	\begin{eqnarray*}
	(B_1+B_2)^{-1/2}(A_1+A_2)(B_1+B_2)^{-1/2} & = & B_1^{-1/2}A_1B_1^{-1/2}\\
	 & = & B_2^{-1/2}A_2B_2^{-1/2}.
	\end{eqnarray*}
	We are thankful to F. Hiai for making us aware of these results.
	
	\section{Barycentres}
	
	If $f$ is a convex function on an open convex set,
	then a critical point of $f$ is the global minimum of $f.$
	If $f$ is strictly convex, then $f$ can have at most one such critical point.
	In this section we show that for $d=d_3$ and $d_4,$
	the objective function in \eqref{eq13} has a critical point,
	and hence in both cases the problem \eqref{eq13} has a unique solution.
	
\begin{thm}\label{thm6}
	When $d=d_3,$ the minimum in \eqref{eq13} is attained at a unique point $X$ which is the solution of the matrix equation
	\eqref{eq17r}
\begin{equation*}
	X^2=\frac{2}{\pi}\sum\limits_{j=1}^{m}w_j\int\limits_{0}^{\infty}\left(\lambda X^{-1}+A_j^{-1}\right)^{-2}\sqrt{\lambda}
	\textrm{d}\lambda.
	\end{equation*}
	This minimiser is the $1/2$-power mean $Q_{1/2}$ given by \eqref{eq13a}
	if $Q_{1/2}$ commutes with every $A_j.$
	In particular, the minimiser is $Q_{1/2}$ if
	\begin{itemize}
	\item[(i)] all $A_j$'s commute, or
	\item[(ii)] $Q_{1/2}=I.$
	\end{itemize}
	\end{thm}
	
	\begin{proof}
	For a fixed positive definite matrix $A,$ define the map $G_A$ as
	$$G_A(X)=A\# X.$$
	By Proposition \ref{prop-new2}, we have
	$$DG_A(X)(Y)=\int\limits_{0}^{\infty}(\lambda+XA^{-1})^{-1}Y(\lambda+A^{-1}X)^{-1}\d\nu(\lambda).$$
	
	The objective function in \eqref{eq13} is 
	$$f(X)=\sum\limits_{j=1}^{m}w_j\Phi_3(X,A_j).$$
	Using the definition of $\Phi_3$ we have
	$$Df(X)(Y)=\tr\left(Y-2\sum\limits_{j=1}^{m}w_jDG_{A_j}(X)(Y)\right).$$
	Then using the above expression for $DG_{A_j}(X)$  we see that 
	\begin{eqnarray*}
	Df(X)(Y) & = & \tr\left(Y-2\sum\limits_{j=1}^{m}w_j\int\limits_{0}^{\infty}(\lambda+XA_j^{-1})^{-1}Y(\lambda+A_j^{-1}X)^{-1}\d\nu(\lambda)\right)\\
	 & = & \tr\left(\left(I-2\sum\limits_{j=1}^{m}w_j\int\limits_{0}^{\infty}\left( (\lambda+XA_j^{-1})(\lambda+A_j^{-1}X)\right)^{-1}\d\nu(\lambda)\right)Y\right).
	\end{eqnarray*}
	At the last step above we use the cyclicity of the trace function.
Hence the critical point of $f$ is the matrix $X_0$ if and only if $X_0$ satisfies the matrix equation
\begin{equation}
I=2\sum\limits_{j=1}^{m}w_j\int\limits_{0}^{\infty}\left((\lambda+XA_j^{-1})(\lambda+A_j^{-1}X)\right)^{-1}\d\nu(\lambda).\label{eq17br}
\end{equation}
Taking congruence with $X$ on both sides we see that \eqref{eq17br} is equivalent to \eqref{eq17r}.
\vskip.1in
We now show that there exists a positive definite matrix $X_0$ that satisfies \eqref{eq17r}.
Let $\alpha,\beta>0$ such that $\alpha I\le A_j\le \beta I$ for all $j=1,\ldots,m,$
and let $\mathcal{K}$ be the compact set $\mathcal{K}=\{X\in\mathbb{P}(n):\alpha I\le X\le \beta I\}.$
Define the map $F:\mathcal{K}\to\mathbb{P}(n)$ as
$$F(X)=\left[2\sum\limits_{j=1}^{m}w_j\int\limits_{0}^{\infty}(\lambda X^{-1}+A_j^{-1})^{-2}\d\nu(\lambda)\right]^{1/2}.$$
Since $X,A_j\in\mathcal{K},$
$(\lambda+1)\alpha^{-1}\ge (\lambda X^{-1}+A_j^{-1})\ge (\lambda+1)\beta^{-1}.$
Thus we have $\alpha^2/(\lambda+1)^2\le (\lambda X^{-1}+A_j^{-1})^{-2}\le \beta^2/(\lambda+1)^2.$
We know that $\int_{0}^{\infty}\d\nu(\lambda)/(\lambda+1)^2=1/2.$
This gives $F(X)\in\mathcal{K}.$
By the Brouwer fixed point theorem,
we get that $F$ has a fixed point $X_0$ in $\mathcal{K}.$
This fixed point $X_0$ is the solution of \eqref{eq17r}.
\vskip.1in
Suppose $Q_{1/2}$ commutes with every $A_j,$ $1\le j\le m.$
	We show that $Q_{1/2}$ satisfies \eqref{eq17r}.
		Differentiating \eqref{eq2b} we get
	\begin{equation}
	\frac{1}{2}x^{-1/2}=\int\limits_{0}^{\infty}\frac{1}{(\lambda+x)^2}\d\nu(\lambda).\label{eq17cr}
	\end{equation}
	Using $Q_{1/2}A_j^{-1}=A_j^{-1}Q_{1/2}$ in \eqref{eq17br} and using \eqref{eq17cr}
	we get
\begin{eqnarray*}
I &=& Q_{1/2}^{1/2}Q_{1/2}^{-1/2}=\sum\limits_{j=1}^{m}w_j\left(A_j^{1/2}Q_{1/2}^{-1/2}\right)\\
 &= & \sum\limits_{j=1}^{m}w_j\left(A_jQ_{1/2}^{-1}\right)^{1/2}\\
 &=& 2\sum\limits_{j=1}^{m}w_j\int\limits_{0}^{\infty}\left(\lambda+A_j^{-1}Q_{1/2}\right)^{-2}\d\nu(\lambda)\\
&=& 2\sum\limits_{j=1}^{m}w_j\int\limits_{0}^{\infty}\left((\lambda+Q_{1/2}A_j^{-1})(\lambda+A_j^{-1}Q_{1/2})\right)^{-1}\d\nu(\lambda).
\end{eqnarray*}
This proves the second statement of the theorem.
If (i) holds, it follows from \eqref{eq13a} that $Q_{1/2}$ commutes with $A_j$'s.
The same is trivially true if (ii) holds.
\end{proof}

	\begin{thm}\label{thm7}
	When $d=d_4$ the minimum in \eqref{eq13} is attained
	at a unique point $X$ which satisfies the matrix equation \eqref{eq18a}
	\begin{equation*}
	X=\sum\limits_{j=1}^{m}w_j\mathcal{L}(X,A_j).
	\end{equation*}
	\end{thm}
	
	\begin{proof}
	Start with the integral representation
	$$\log x=\int\limits_{0}^{\infty}\left(\frac{\lambda}{\lambda^2+1}-\frac{1}{\lambda+x}\right)\d\lambda,\ x>0.$$
	This shows that for all $X>0$ and all Hermitian $Y$ we have
	$$D(\log X)(Y)=\int\limits_{0}^{\infty}(\lambda+X)^{-1}Y(\lambda+X)^{-1}\d\lambda.$$
	For a fixed $A,$ let
	$$g(X)=\frac{1}{2}(\log A+\log X).$$
	Then
	\begin{equation}
	Dg(X)(Y)=\frac{1}{2}\int\limits_{0}^{\infty}(\lambda+X)^{-1}Y(\lambda+X)^{-1}\d\lambda.\label{eq46}
	\end{equation}
	The log Euclidean mean $\mathcal{L}(A,X)=\e^{g(X)}.$
	So, by the chain rule and Dyson's formula (see \cite{rbh} p. 311), we have
	$$D\mathcal{L}(A,X)(Y)=\int\limits_{0}^{1}\e^{(1-t)g(X)}Dg(X)(Y)\e^{tg(X)}\d t.$$
	This shows that 
	\begin{eqnarray*}
	D(\tr\mathcal{L}(A,X))(Y) & = & \tr\int\limits_{0}^{1}\e^{(1-t)g(X)}Dg(X)(Y)\e^{tg(X)}\d t\\
	 & = & \tr\left[\e^{g(X)}Dg(X)(Y)\right],
	\end{eqnarray*}
	using the cyclicity of trace.
	Using \eqref{eq46} and the cyclicity once again,
	we obtain
	\begin{eqnarray*}
	D(\tr\mathcal{L}(A,X))(Y) & = & \frac{1}{2}\tr\int\limits_{0}^{\infty}(\lambda+X)^{-1}\e^{g(X)}(\lambda+X)^{-1}Y\d\lambda\\
	 & = & \frac{1}{2}\tr\left(\int\limits_{0}^{\infty}(\lambda+X)^{-1}\mathcal{L}(A,X)(\lambda+X)^{-1}\d\lambda\right)Y.
	\end{eqnarray*}
	Hence, for the function 
	$$\Phi_4(A,X)=d_4^2(A,X)=\tr(A+X)-2\tr\mathcal{L}(A,X),$$
	we have
	\begin{align*}
	 & D\Phi_4(A,X)(Y)\\
	 &  = \tr\left( I-\int\limits_{0}^{\infty}(\lambda+X)^{-1}\mathcal{L}(A,X)(\lambda+X)^{-1}\d\lambda\right)Y.
	\end{align*}
	The objective function in \eqref{eq13} is 
	$$f(X)=\sum\limits_{j=1}^{m}w_j\Phi_4(A_j,X).$$
	So, we have
	\begin{align}
	 & Df(X)(Y)\nonumber\\
	 & = \tr\left(I-\int\limits_{0}^{\infty}(\lambda+X)^{-1}Z(\lambda+X)^{-1}\d\lambda\right)Y,\label{eq47}
	\end{align}
	where 
	$$Z=\sum\limits_{j=1}^{m}w_j\mathcal{L}(A_j,X).$$
	This shows that $Df(X)=0$ if and only if 
	\begin{equation}
	\int\limits_{0}^{\infty}(\lambda+X)^{-1}Z(\lambda+X)^{-1}\d\lambda=I.\label{eq48}
	\end{equation}
	Choose an orthonormal basis in which 
	$X=\diag(x_1,\ldots,x_n),$
	and let $Z=\begin{bmatrix}z_{ij}\end{bmatrix}$ in this basis.
	Then the condition \eqref{eq48} says that 
	$$\int\limits_{0}^{\infty}\frac{z_{ij}}{(\lambda+x_i)(\lambda+x_j)}\d\lambda=\delta_{ij} \textrm{  for all }i,j.$$
	This shows that $Z$ is diagonal, and 
	$$\frac{1}{z_{ii}}=\int\limits_{0}^{\infty}\frac{1}{(\lambda+x_i)^2}\d\lambda=\frac{1}{x_i}.$$
	Thus $X=Z=\sum\limits_{j=1}^{m}w_j\mathcal{L}(A_j,X),$ as claimed.
	\vskip.1in
	We should also show that the equation \eqref{eq18a} has a unique solution.
	Let $\alpha,\beta$ be positive numbers such that $\alpha I\le A_j\le \beta I$ for all $1\le j\le m.$
	Let $\mathcal{K}$ be the compact convex set $\mathcal{K}=\{X\in\mathbb{P}:\alpha I\le X\le \beta I\}.$
	The function $\log X$ is operator monotone.
	So for all $X$ in $\mathcal{K}$
	we have $\log\alpha I\le \log X\le \log\beta I.$
	Hence $\mathcal{L}(X,A_j)$ is in $\mathcal{K}$ for all $1\le j\le k.$
	This shows that the function $F(X)=\sum\limits_{j=1}^{m}w_j\mathcal{L}(X,A_j)$
	maps $\mathcal{K}$ into itself.
	By Brouwer's fixed point theorem $F$ has a unique fixed point $X$ in $\mathcal{K}.$
	This $X$ is a solution of \eqref{eq18a}
	and therefore must be unique.
	\end{proof}
	
	Finally, we remark that in the case of $d_1,$ the barycentre is given explicitly by the formula \eqref{eq13a}.
	For $d_2,$ $d_3,$ $d_4$ it has been given implicitly as solution of the equations
	\eqref{eq13b},\eqref{eq17r},\eqref{eq18a}, respectively.
	When $m=2$ and $w_1=w_2=1/2$, \rajendra{$w$ precised}
the solution of \eqref{eq13b} is the Wasserstein mean of $A_1$ and $A_2$ defined as
	$$\frac{1}{4}\left(A_1+A_2+(A_1A_2)^{1/2}+(A_2A_1)^{1/2}\right).$$
	See \cite{bjl}.
	\vskip.2in
	
	{\bf Acknowledgements}: The authors thank F. Hiai and S. Sra for helpful comments and references,
	and the anonymous referee for a careful reading of the manuscript.
	The first author is grateful to INRIA and \'Ecole polytechnique, Palaiseau
	for visits that facilitated this work, and to CSIR(India) for the award of a Bhatnagar Fellowship.
	
	\vskip.3in

\appendix 
\section{Proof of \Cref{lem-new}}\label{appendixa}
We make a variation of the proof of Theorem 3.12 in \cite{bb97}, 
dealing with a related problem (the minimisation of $\Phi$ over a closed convex set). 

Since $\varphi$ is of Legendre type, Theorem 3.7(iii) of \cite{bb97} shows that
for all $a\in\interior\dom \varphi$, the map $x\mapsto \Phi(x,a)$ is coercive, meaning that $\lim_{\|x\|\to\infty} \Phi(x,a)=+\infty$. A sum of coercive functions is coercive, and so the map 
\[
\Psi(x):= \sum_{j=1}^m\frac{1}{m}\Phi(x,a_j)
\]
is coercive. The infimum of a coercive lower-semicontinuous function on a closed non-empty set is attained, so there is an element
$\bar{x}\in \closure\interior \dom \varphi$ such that $\inf_{x\in \closure\interior \dom \varphi} \Phi(x)=\Phi(\bar{x})<+\infty$. 
Suppose that $\bar{x}$ belongs to the boundary of $\interior\dom \varphi$.
Let us fix an arbitrary $z\in\interior\dom \varphi$, and
let $g(t) :=\Psi((1-t)\bar{x}+t z)$, defined for $t\in [0,1)$. We have
\[
g'(t)=\langle \nabla \varphi((1-t)\bar{x}+t z)-\sum_{j=1}^m \frac{1}{m}
\nabla\varphi(a_j), z-\bar{x}\rangle \enspace .
\]
Using property~\eqref{itdef-iv} of the definition of Legendre type functions,
we get that $\lim_{t\to 0^+}g'(t)=-\infty$, which entails that $g(t)<g(0)=\Psi(\bar x)$ for $t$ small enough. Since $(1-t)\bar{x}+t z\in \interior\dom \varphi$
for all $t\in(0,1)$, this contradicts the optimality of $\bar x$. So
$\bar x\in \interior\dom \varphi$, which proves \Cref{lem-new}.

\section{Examples}\label{appendixb}
In the last statement of \Cref{thm3}, dealing with tracial convex
functions, we required $\varphi$ to be differentiable
and strictly convex on $\mathbb{P}$.
In the second statement, dealing with the non tracial case,
we made a stronger assumption, requiring $\varphi$ to be of Legendre
type. We now give an example showing that the Legendre condition cannot be dispensed with.
To this end, it is convenient to construct first an example 
showing the tightness of~\Cref{lem-new}.

\subsection*{Need for the Legendre condition in \Cref{lem-new}}

Let us fix $N>3$, let $e=(1,1)^\top\in\R^2$, 
\begin{align}
L=\left(\begin{array}{cc}
N-1 & -2 \\ -2 & N-1
\end{array}\right)\label{e-def-L}
\end{align}
and consider the affine transformation $g(x)=e+Lx$. 
Let $ a =
(N,0)^\top$,
$ b = 
(0 , N)^\top$,
and
\[
\bar a:= g^{-1}( a)= \frac{1}{N^2-2N-3}\left(\begin{array}{c} N^2-2N-1 \\ N-1
\end{array}
\right),
\]
\[
\bar b:=g^{-1}(b)= \frac{1}{N^2-2N-3}
\left(\begin{array}{c}N-1 \\ N^2-2N-1
\end{array}
\right) \enspace .
\]
Observe that $\bar a, \bar b\in\positivereals^2$ since $N>3$. 

Consider now, for $p>1$, the map $\varphi(x):=\|x\|_p^p=|x_1|^p+|x_2|^p$ defined
on $\R^2$ and $\bar\varphi(x)=\varphi(g(x))$. Observe that $\varphi$
is strictly convex and differentiable. Let $\bar\Phi$
denote the Bregman divergence associated with $\bar\varphi$,
and let $\bar\Psi(x):= \frac{1}{2}(\bar\Phi(x,\bar a)+\bar\Phi(x,\bar b))$. We claim that
$0$ is the unique point of minimum of $\bar\Psi$ over $\nonnegativereals^2$.
Indeed, 
\begin{align*}
\nabla\bar\Psi(x)&=L^\top(\nabla\varphi(g(x)))-\frac{1}2
\Big(L^\top(\nabla\varphi( a))+
L^\top(\nabla\varphi( b))\Big)\enspace ,
\end{align*}
from which we get
\begin{align*}
\nabla\bar\Psi(0)&= L(p(1 -N^{p-1}/2) e)=(N-3)p(1-N^{p-1}/2)e \enspace .
\end{align*}
It follows that $\nabla\bar\Psi(0)\in \positivereals^2$ 
if $p>1$ is chosen close enough to $1$, so that $1-N^{p-1}/2>0$.
Then, since $\bar\Psi$ is convex, we have
\begin{align}
\bar\Psi(x)-\bar\Psi(0)\geq \langle \nabla \bar\Psi(0),x\rangle>0,
\qquad \text{ for all } x\in \nonnegativereals^2\setminus\{0\}
\label{e-strict}
\end{align}
showing the claim. 

Consider now the modification $\hat{\varphi}$ of $\bar{\varphi}$, 
so that $\hat{\varphi}(x)=\bar{\varphi}(x)$ for $x\in \nonnegativereals^2$,
and $\hat{\varphi}(x)=+\infty$ otherwise. The function $\hat{\varphi}$
is strictly convex, lower-semicontinuous, and differentiable on the interior of its domain, but not of Legendre type, and the conclusion of \Cref{lem-new} does not apply to it. 

The geometric intuition leading to this example is described in the figure.%
\stephane{added last sentence}
\begin{figure}
\begin{center}
\begin{tikzpicture}[scale=6.5,convex/.style={fill=lightgray,fill opacity=0.50},convexborder/.style={very thick}]
\node at (0.055,0.055) {
\includegraphics[width=19.5mm,height=19.5mm]{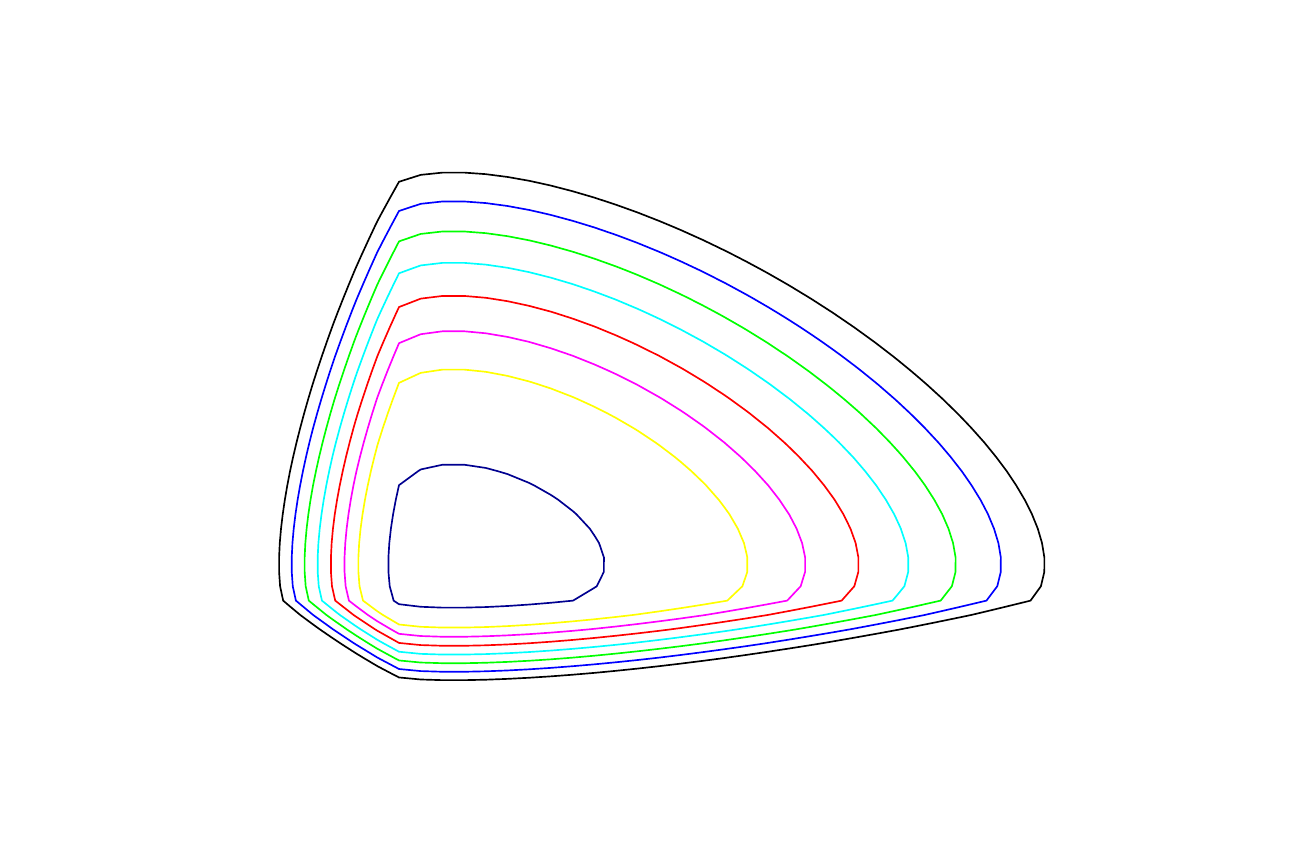}
};
\draw[gray!60, ultra thin,step=0.1] (-0.1,-0.1) grid (0.5,0.5);
\draw (0.4,0) node[right] {$a$};
\filldraw (0.4,0) circle (0.06ex);
\draw (0,0.4) node[right] {$b$};
\filldraw (0,0.4) circle (0.06ex);
\draw (0.1,0.1) node[right] {$e$};
\filldraw (0.0125,0.0125) circle (0.06ex);
\draw (0.0125,0.0125) node[right] {$u$};
\draw (0.3,0.3) node[right] {$C$};
\filldraw (0.1,0.1) circle (0.06ex);
\filldraw[gray,opacity=0.3] (0.5,-0.333333*0.1) -- (0.1,0.1) -- (-0.333333*0.1,0.5) -- (0.5,0.5)-- cycle;
\end{tikzpicture}
\end{center}\captionsetup{labelformat=empty}
\caption{The example illustrated. The point $u$ is the unconstrained minimum of the sum of Bregman divergences $\Psi(x):=\Phi(x,a)+\Phi(x,b)$ associated with $\varphi(x)=x_1^p+x_2^p$, here $p=1.2$. Level curves of $\Psi$ are shown. The minimum of $\Psi$ on the simplicial cone $C$ is at the unit vector $e$. An affine change of variables sending $C$ to the standard quadrant, and a lift to the cone of positive semidefinite matrices leads to \Cref{th-cex}}\label{fig-explain}
\end{figure}

\subsection*{Need for the Legendre condition in \Cref{thm3}}
We next construct an example showing that the Legendre
condition in the second statement of \Cref{thm3} cannot be dispensed
with. Observe that the inverse of the linear operator $L$ in~\eqref{e-def-L}
is given by
\[
L^{-1} = 
\frac{1}{N^2-2N-3}\left(\begin{array}{cc}
N-1 & 2 \\ 2 & N-1
\end{array}\right) \enspace.
\]
In particular, it is a nonnegative matrix. 

We set $\tau=\left(\begin{smallmatrix}0&1\\1& 0\end{smallmatrix}\right)$,
and consider the ``quantum'' analogue of $L$, i.e.,
\[
T(X)=(N-1)X- 2 \tau X \tau \enspace .
\]
Then,
\[
T^{-1}(X)= \frac{1}{N^2-2N-3}\big((N-1)X+ 2\tau X\tau\big)
\]
is a completely positive map leaving $\mathbb{P}$
invariant. The analogue of the map $g$ is
\[
G(X) =I + T(X)
\]
where $I$ denotes the identity matrix. 

We now consider the map $\varphi(X):= \|X\|_p^p=\operatorname{tr}(|X|^p)$
defined on the space of Hermitian matrices. The function $\varphi$ is differentiable and strictly convex, still assuming that $p>1$.
We set $\bar A:= \diag(\bar a)\in \mathbb{P}$, $\bar B:=\diag(\bar b)\in \mathbb{P}$, and now define $\bar\Phi$ to be the Bregman divergence associated
with $\bar\varphi:= \varphi\circ G$. Let
\[
\bar\Psi(X):= \frac{1}{2}\Big(
\bar\Phi(X,\bar A)+ \bar\Phi(X,\bar B) 
\Big) \enspace .
\]
We then have the following result.
\begin{prop}\label{th-cex}
The minimum of the function $\bar\Psi$ on the closure of $\mathbb{P}$ is achieved at point $0$. Moreover, the equation
\begin{align}
\nabla \bar{\varphi}(X)=\frac{1}{2}(\nabla \bar{\varphi}(\bar{A})
+ \nabla\bar{\varphi}(\bar{B}))\label{e-nosolution}
\end{align}
has no solution $X$ in $\mathbb{P}$. 
\end{prop}
\begin{proof}
From~\cite{aikenerdosgoldstein} (Theorem 2.1) or \cite{abatzoglou} (Theorem 2.3),
we have
\[
\frac{d}{dt}\mid_{t=0} \operatorname{tr}|X+tY|^p = p \operatorname{Re} \operatorname{tr}|X|^{p-1}U^*Y
\]
where $X=U|X|$ is the polar decomposition of $X$. In particular, 
if $X$ is diagonal and positive semidefinite,
\[
\nabla \varphi(X) = pX^{p-1} \enspace .
\]
Then, by a computation similar to the one in the scalar case above,
we get
\[
\nabla \bar\Psi(0) = (N-3)p(1-N^{p-1}/2)I \in \mathbb{P} \enspace.
\]
We conclude, as in~\eqref{e-strict}, that 
\[
\bar\Psi(X)-\bar\Psi(0)\geq \langle \nabla \bar\Psi(0),X\rangle>0,
\qquad \text{ for all } X\in \closure \mathbb{P}\setminus\{0\} \enspace,
\]
where now $\langle \cdot,\cdot\rangle$ is the Frobenius scalar
product on the space of Hermitian matrices.
It follows that $0$ is the unique point of minimum of $\bar\Psi$
on $\closure\mathbb{P}$. 

Moreover, if the equation~\eqref{e-nosolution} had a solution
$X\in\mathbb{P}$, the first order optimality condition
for the minimisation of the function $\bar\Psi$ over
$\mathbb{P}$ would be satisfied, showing that $\bar\Psi(Y)\geq \bar\Psi(X)$
for all $X\in \mathbb{P}$, and by density, $\bar\Psi(0)\geq \bar\Psi(X)$,
contradicting the fact that $0$ is the unique point of minimum of $\bar\Psi$
over $\closure \mathbb{P}$.
\end{proof}
\vskip.2in

\noindent{\it\bf Note added to the second version}:
In the earlier version of this paper posted on January 5, 2019
that appeared in Letters in Mathematical Physics, 109, (2019) 1777-1804, ,
we made an unfortunate error.
Theorem 9 in that version wrongly claimed that for the case $d=d_3$
the solution of the minimisation problem \eqref{eq13}
is also the solution of the matrix equation \eqref{eq13r}.
The mistake in the statement and in the proof has been pointed in
J. Pitrik and D. Virosztek, Quantum Hellinger distances revisited, arXiv: 1903.10455v3.
In this paper some more general divergence functions are considered,
the barycentre equations are derived,
and an example is given to show that
the solution to the matrix equations \eqref{eq17r} and \eqref{eq13r}
need not be the same.

\todo{fix the date.}
\hfill{January 4, 2019}

\end{document}